\RequirePackage{fix-cm}
\documentclass[smallextended,eqsecnum]{svjour3}       
\smartqed  
\def\kF{k_{\rm{F}}}

\def\epsilonF{\epsilon_{\rm{F}}}

\def\sgn{{\rm{sgn\,}}}

\def\be{\begin{equation}}
\def\ee{\end{equation}}
\def\bea{\begin{eqnarray}}
\def\eea{\end{eqnarray}}
\def\bse{\begin{subequations}}
\def\ese{\end{subequations}}

\def\intR{\int_{\mathbbm{R}}}

\usepackage[T1]{fontenc}
\usepackage[latin9]{inputenc}
\usepackage{amsmath}
\numberwithin{equation}{section}
\numberwithin{remark}{section}
\numberwithin{lemma}{section}
\numberwithin{proposition}{section}
\numberwithin{theorem}{section}
\numberwithin{corollary}{section}
\usepackage{graphicx}
\usepackage{rotating}
\usepackage {bm}
\usepackage[section]{placeins}
\usepackage{color}
\usepackage{multirow}
\usepackage{cellspace}
\usepackage{bbm}
\journalname{Communications in Mathematical Physics}
\begin{document}

\title{Rigorous results for the electrical conductivity due to electron-phonon scattering
}



\author{J. Amarel         \and
             D. Belitz          \and
             T.~R. Kirkpatrick
}


\institute{J. Amarel \at
              Dept. of Physics and Institute for Fundamental Science \\
              University of Oregon\\
              Eugene, OR 97403, USA\\
              \email{jamarel@uoregon.edu}           
           \and
           D. Belitz \at
              Dept. of Physics and Institute for Fundamental Science 
              and Materials Science Institute\\
              University of Oregon\\
              Eugene, OR 97403, USA\\
              \email{dbelitz@uoregon.edu}                 
            \and
           T.~R. Kirkpatrick \at
              Institute for Physical Science and Technology\\
              University of Maryland\\
              College Park, MD 20742, USA\\
              \email{tedkirkp@umd.edu}
}

\date{Received: date / Accepted: date}

\maketitle

\begin{abstract}
We present a rigorous solution of the Boltzmann equation for the electron-phonon scattering
problem in three spatial dimensions in the limit of low temperatures. The different temperature scaling of the various
scattering rates turns the temperature into a control parameter that is not available in
classical kinetic theory and allows for a rigorous proof of Bloch's $T^5$ law. The relation
between the Boltzmann equation and the Kubo formula is also discussed, as well as
implications for the scattering of electrons by excitations other than phonons.
\end{abstract}

\tableofcontents

\section{Introduction}
\label{sec:1}

\subsection{The electron-phonon problem}
\label{subsec:1.1}

Consider the scattering of conduction electrons in a bulk metal by excitations with a momentum-dependent 
resonance frequency $\omega_{\bm q}$; for definiteness, we will mainly consider long-wavelength acoustic phonons, in which
case $\omega_{\bm q} = c\vert{\bm q}\vert$ with $c$ the speed of sound. The scattering leads, for any nonzero
temperature $T$, to a finite electrical conductivity $\sigma$ that is given by a Drude formula
\bse
\label{1.1}
\be
\sigma(T) = n\,e^2\,\tau(T)/m
\label{1.1a}
\ee
with $n$, $e$, and $m$ the electronic number density, charge, and mass, respectively. The transport relaxation time
$\tau$ is given by the energy average of an inverse rate $\varphi$,
\be
\tau(T) = \frac{1}{2T} \intR d\epsilon\,w(\epsilon)\,\varphi(\epsilon)\ .
\label{1.1b}
\ee
\ese
Here the weight function
\be
w(\epsilon)  = f_0(\epsilon/T)\left[1 - f_0(\epsilon/T)\right] = -T\,\frac{\partial f_0(\epsilon/T)}{\partial\epsilon} = \frac{1}{4\cosh^2(\epsilon/2T)}
\label{1.2}
\ee
is defined in terms of the Fermi distribution function
\bse
\label{1.3}
\be
f_0(x) = 1/(e^{x}+1) \ .
\label{1.3a}
\ee
The corresponding distribution function for phonons in equilibrium is the Bose distribution
\be
n_0(x) = 1/(e^{x} - 1)\ .
\label{1.3b}
\ee
\ese
A standard way to determine the rate $\varphi$ is to start with the Boltzmann equation for the fermion distribution function $f$.
The deviation from the equilibrium distribution $f_0$ due to an external electric field ${\bm E}$ can be parameterized by a function $\Phi$,
\bse
\label{1.4}
\be
f({\bm k}) = f_0(\epsilon_{\bm k}/T) -\frac{\partial f_0(\epsilon_{\bm k}/T)}{\partial\epsilon_{\bm k}}\,\Phi({\bm k})\ .
\label{1.4a}
\ee
For calculating the conductivity it is further convenient to write
\be
\Phi({\bm k}) = (e/m){\bm E}\cdot{\bm k}\,\varphi(\epsilon_{\bm k})\ .
\label{1.4b}
\ee
\ese
Here ${\bm k} \in {\mathbbm{R}}^3$ labels the momentum of a single-electron state with energy $\epsilon_{\bm k}$, and $\varphi$ is the inverse rate
that determines the transport scattering time via \eqref{1.1b}.
\begin{remark}
For simplicity we assume a parabolic band, $\epsilon_{\bm k} = {\bm k}^2/2m$, and we use units such that $\hbar = k_{\text{B}} = 1$.
\label{remark:1.1}
\end{remark}
The Boltzmann equation reads \cite{Wilson_1954,Ziman_1960}
\be
-e{\bm E}\cdot{\partial_{\bm k}}f({\bm k}) = (\partial f/\partial t)_{\text{coll}}({\bm k})\ .
\label{1.5}
\ee
The collision operator on the right-hand side describes the change of the distribution function per unit time due to the scattering
by the phonons and balances the streaming term on the left-hand side. To linear order in $\Phi$, which suffices for calculating
the conductivity, it reads (\cite{Wilson_1954} Secs. $8\!\cdot\! 1$ and $9\!\cdot\! 34$)
\bse
\label{1.6}
\be
(\partial f/\partial t)_{\text{coll}}^{\text{lin}}({\bm k}) \equiv (C f)({\bm k}) = \frac{1}{\kF^3} \int_{\mathbbm{R}^3} d{\bm k}'\ W({\bm k},{\bm k}')\,\left[\Phi({\bm k}) - \Phi({\bm k}')\right]\ .
\label{1.6a}
\ee
Omitting a numerical prefactor of order unity, the kernel $W$ reads
\bea
W({\bm k},{\bm k}') &=& \frac{\omega_0 \epsilonF}{T} \frac{\vert {\bm k} - {\bm k}'\vert^2}{2m} \frac{n_0(\omega_{{\bm k}-{\bm k}'}/T)}{\omega_{{\bm k}-{\bm k}'}}
\nonumber\\
&& \times
     \left[f_0(\epsilon_{\bm k}/T)\left(1 - f_0(\epsilon_{{\bm k}'}/T)\right) \delta\left(\epsilon_{\bm k} - \epsilon_{{\bm k}'} + \omega_{{\bm k}-{\bm k}'}\right) \right.  
\nonumber\\
   && \hskip 10pt + \left. f_0(\epsilon_{{\bm k}'}/T)\left(1 - f_0(\epsilon_{\bm k}/T)\right) \delta\left(\epsilon_{\bm k} - \epsilon_{{\bm k}'} - \omega_{{\bm k}-{\bm k}'}\right) \right]\ .
\label{1.6b}
\eea
\ese
Here $\kF$ and $\epsilonF$ are the Fermi wave number and energy, respectively, and $\omega_0 = c\kF$ is the characteristic phonon
energy scale, which is on the order of the Debye temperature. 
\begin{remark}
This treatment of the electron-phonon problem assumes that the phonons remain in thermal
equilibrium. This (unjustifiable) approximation goes back to Bloch (see \cite{Ziman_1960} p. 358). 
A treatment of the complete coupled system is
more complicated \cite{Ziman_1960,Ashcroft_Mermin_1976}.
\label{remark:1.2}
\end{remark}
\begin{remark}
The kernel $W$ is symmetric, and hence the collision operator $C$ is self-adjoint in the space $L^2$ of square-integrable
functions.
\label{remark:1.3}
\end{remark}
\begin{remark}
$C$ has an obvious zero eigenvalue. The corresponding eigenfunction is given by the distribution function with $\Phi({\bm k}) \equiv \text{const.}$
This reflects electron number conservation.
\label{remark:1.4}
\end{remark}
\begin{remark}
The electron momentum is {\em not} conserved, since the phonons can absorb momentum. Accordingly, if $C$ acts on $\Phi$
as written in \eqref{1.4b} with $\varphi(\epsilon_{\bm k}) = \text{const.}$, the result is not zero. However, since $\vert{\bm k}-{\bm k}'\vert$
scales as the temperature (for phonons, or as some positive power of the temperature in general) due to the Bose distribution function 
in \eqref{1.6b}, electron momentum is approximately conserved in the limit $T\to 0$. This will be important later.
\label{remark:1.5}
\end{remark}

\subsection{An integral equation for the scattering rate}
\label{subsec:1.2}

\subsubsection{Scattering by a generic potential}
\label{subsubsec:1.2.1}

For a calculation of the electrical conductivity it is convenient to recast the linearized Boltzmann equation
\eqref{1.5}, \eqref{1.6} in the form of an integral equation for the inverse rate $\varphi$ \cite{Wilson_1954}. At this point it is
easy to generalize the problem to the scattering of the conduction electrons by a given dynamical potential 
$V({\bm p},z)$, with $z$ a complex frequency, and $V''({\bm p},u) = \lim_{\delta\to 0} [V({\bm p},u+i\delta) - V({\bm p},u-i\delta)]/2i$ 
the spectrum of the potential. Inserting \eqref{1.4b} into \eqref{1.6a}, and performing the angular integration, leads to an integral equation
\be
\varphi(\epsilon) = 1/\Gamma_0(\epsilon) + \intR du\,K(\epsilon,u)\,\varphi(u)\ .
\label{1.7}
\ee
The kernel $K$ consists of three contributions,
\be
K(\epsilon,u) = K_0(\epsilon,u) + K_1(\epsilon,u) - K_2(\epsilon,u) 
\label{1.8}
\ee
that are defined as
\bse
\label{1.9}
\be
K_{\nu}(\epsilon,u) = {\bar K}_{\nu}(\epsilon,u)/\Gamma_0(\epsilon) \qquad (\nu=0,1,2)
\label{1.9a}
\ee
where
\be
{\bar K}_{\nu}(\epsilon,u) = \left[ n\left(\frac{u-\epsilon}{T}\right) + f\left(\frac{u}{T}\right)\right] \, {\bar V}''_{\nu}(u-\epsilon)
\label{1.9b}
\ee
\ese
with
\bse
\label{1.10}
\bea
{\bar V}''_0(u) &=& \frac{1}{2\kF^2} \int_0^{2\kF} dp\,p\, V''({\bf p},u)
\label{1.10a}\\
{\bar V}''_1(u) &=& \frac{u}{2\epsilonF}\,{\bar V}''_0(u)
\label{1.10b}\\
{\bar V}''_2(u) &=& \frac{1}{2\kF^2} \int_0^{2\kF} dp\,p\, (p^2/2\kF^2)\,V''({\bf p},u)
\label{1.10c}
\eea
\ese
Here we have normalized the potential by means of the electronic density of states such that the ${\bar V}_{\nu}''$ 
are dimensionless. The function $\Gamma_0$ in \eqref{1.7}, \eqref{1.9a} is defined as
\be
\Gamma_0(\epsilon) = \intR du\,{\bar K}_0(\epsilon,u)\ .
\label{1.11}
\ee
\begin{remark} 
Physically, $\Gamma_0$ is the single-particle relaxation rate.
\label{remark:1.6}
\end{remark}
\begin{remark}
Physically, the potential $V$ describes an effective elelectron-electron interaction that is mediated
by the exchange of bosonic excitations that have been integrated out. ${\bar V}_{0,2}''$ are 
odd functions of their argument, ${\bar V}''_{0,2}(-u) = -{\bar V}''_{0,2}(u)$, while
${\bar V}''_1(-u) = {\bar V}''_1(u)$ is even.  As a result, $K_{0,2}$ are positive semi-definite
with $K_{0,2}(-\epsilon,-u) = K_{0,2}(\epsilon,u)$, while $K_1(-\epsilon,-u) = - K_1(\epsilon,u)$.
\label{remark:1.7}
\end{remark}

\subsubsection{Electron-phonon scattering}
\label{subsubsec:1.2.2}

Specializing to phonons again, the basic potential spectrum defined in \eqref{1.10a} becomes
\be
{\bar V}_0''(u) \propto(u/\omega_0)^2\ \sgn u
\label{1.12}
\ee
and absorbing a numerical prefactor into $K_0$ we have
\bse
\label{1.13}
\be
{\bar K}_0(\epsilon,u) = \left(\frac{T}{\omega_0}\right)^2 {\bar k}_0\left(\frac{\epsilon}{T},\frac{u-\epsilon}{T}\right)
\label{1.13a}
\ee
with
\be
{\bar k}_0(x,y) = \left[ n(y) + f(y+x)\right] y^2\,\sgn (y)\ .
\label{1.13b}
\ee
The other two kernels are
\bea
{\bar K}_1(\epsilon,u) &=& \frac{\omega_0}{2\epsilonF}\,\frac{T}{\omega_0}\,\frac{u-\epsilon}{T}\,{\bar K_0}(\epsilon,u)
\label{1.13c}\\
{\bar K}_2(\epsilon,u) &=& \frac{1}{2} \left(\frac{T}{\omega_0}\right)^2\,\left(\frac{u-\epsilon}{T}\right)^2\,{\bar K}_0(\epsilon,u)\ .
\label{1.13d}
\eea
\ese
The function $\Gamma_0$ from \eqref{1.11} now has the form
\bse
\label{1.14}
\be
\Gamma_0(\epsilon) = (T^3/\omega_0^2)\,\gamma_0(\epsilon/T)
\label{1.14a}
\ee
and for future reference we define
\bea
\Gamma_1(\epsilon) &=& \intR du\,{\bar K}_1(\epsilon,u) = (T^4/\omega_0^2\epsilonF)\,\gamma_1(\epsilon/T)
\label{1.14b}\\
\Gamma_2(\epsilon) &=& \intR du\,{\bar K}_2(\epsilon,u) = (T^5/\omega_0^4)\,\gamma_2(\epsilon/T)
\label{1.14c}
\eea
where
\be
\gamma_{\nu}(\epsilon) = \intR du\,u^{\nu}\,{\bar k}_0(\epsilon,u) \quad (\nu=0,1,2)\ .
\label{1.14d}
\ee
\ese
The transport problem is now completely defined and consists of solving the integral equation \eqref{1.7}
in the limit $T\to 0$. This has been done in various approximations that yield the well-known Bloch law 
$\sigma(T\to 0) \propto T^{-5}$ \cite{Bloch_1930,Wilson_1954,Ziman_1960,Mahan_2000}.
Here we provide a rigorous proof of the Bloch law.  A sketch of a proof was given in \cite{Belitz_Kirkpatrick_2010a}, 
but the treatment was not rigorous due to various hidden assumptions.
\begin{remark}
\eqref{1.7} with the kernels given by \eqref{1.13} is identical to \cite{Wilson_1954} (9$\cdot$61$\cdot$1)
except for terms that are exponentially small of order $\exp(-\omega_0/T)$. 
As outlined above, it represents the linearized Boltzmann equation for the electron-phonon scattering
problem in an approximation that treats the phonons as staying in equilibrium. 
\label{remark:1.8}
\end{remark}
\begin{remark}
Alternatively, the same equation can be derived from the Kubo formula \cite{Kubo_1957} which
expresses $\sigma$ as the current-current Kubo function at zero wave number in the limit of zero frequency.
We recall that the full, nonlinear, Boltzmann equation is exact to linear order in the scattering cross-section and to all orders
in the external field, and the linearized Boltzmann equation is exact to linear order in both the scattering cross-section
and the external field. The Kubo formula is exact to linear order in the external field and to all orders in the scattering
potential. That is, it provides an exact expression for the transport coefficient. If one evaluates it to linear order
in the scattering cross-section, one therefore recovers the linearized Boltzmann equation \cite{Holstein_1964}.
\label{remark:1.9}
\end{remark}
\begin{remark}
The Boltzmann equation is exact for fixed time (or frequen\-cy/tem\-perature) in the limit of a vanishing 
scattering cross-section, but not for a fixed cross-section, no matter how small, in the limit of long
times or small frequency/temperature. A manifestation of this is the existence of long-time tails in
equilibrium time correlation functions \cite{Alder_Wainwright_1970,Dorfman_Cohen_1970,Ernst_Hauge_van_Leeuwen_1970}
in both classical and quantum fluids \cite{Kirkpatrick_Belitz_Sengers_2002}. A rigorous solution of the
Boltzmann equation therefore does not necessarily provide the exact behavior of the transport
coefficient under consideration.
\label{remark:1.9.1}
\end{remark}
\begin{remark}
The equivalence between the Boltzmann equation and the Kubo formula mentioned in Remark~\ref{remark:1.9} is violated by derivations that
neglect the kernel $K_1$ (e.g., \cite{Takegahara_Wang_1977,Hansch_Mahan_1983}), an omission that is sometimes based on the
argument that $K_1$ is suppressed by the small prefactor $\omega_0/2\epsilonF$ in \eqref{1.13c}.  This argument is {\em not} valid a priori.
The leading temperature dependence of $\sigma$ in the absence of $K_1$ is determined by $K_2$,
and the temperature scaling apparent in \eqref{1.13} shows that $K_1$ has the potential for providing
the leading contribution to $\sigma$ for $T \ll \omega_0^2/\epsilonF$. As we will see, for phonons it actually provides a
contribution with the same temperature scaling as the one coming from $K_2$, but this requires a proof.
For other excitations it can provide the leading temperature scaling, see Remark~\ref{remark:4.6}.
\label{remark:1.10}
\end{remark}
\begin{remark}
The different temperature scaling of the kernels $K_{1,2,3}$ that is apparent from \eqref{1.13} allows for
a rigorous solution of \eqref{1.7} in the limit $T\to 0$, as we will show in Sec.~\ref{sec:3}. This is a
crucial difference between the quantum kinetic problem we are studying and classical kinetic theory,
where no analogous limit exists.
\label{remark:1.11}
\end{remark}

\section{A formal solution of the integral equation}
\label{sec:2}

Here we provide a formal solution of the integral equation \eqref{1.7} with the kernels given by \eqref{1.13} 
that leads to Bloch's $T^5$ law. The rigorous solution in Sec.~\ref{sec:3} will follow the same logic.

We start by stating basic symmetry relations of the kernels that we will also use repeatedly in Sec.~\ref{sec:3}.
They can be verified by explicit elementary calculations.
\begin{lemma}
The kernels obey
\bse
\label{2.1}
\bea
w(\epsilon) {\bar K}_{0,2}(\epsilon,u) &=& w(u) {\bar K}_{0,2}(u,\epsilon)\ ,
\label{2.1a}\\
w(\epsilon) {\bar K}_{1}(\epsilon,u) &=& -w(u) {\bar K}_{1}(u,\epsilon)\ .
\label{2.1b}
\eea
\ese
\label{lemma:2.1}
\end{lemma}
We now turn to \eqref{1.7}, which is a Fredholm integral equation of the second kind. It has the form
\be
\bar{\cal C}\vert\varphi) = - \vert 1)\ .
\label{2.2}
\ee
$\bar{\cal C} = {\bar K} - \Gamma_0 \mathbbm{1}$ is a collision operator in the Hilbert space of functions
that are square integrable with respect to the weight $w$; i.e., the scalar product is $(\varphi\vert\psi) = \intR w(\epsilon)\, \varphi(\epsilon)\,\psi(\epsilon)$.
$\mathbbm{1}$ is the identity operator, 
$\vert\varphi)$ is a vector in the Hilbert space that is represented by the unknown function $\varphi$,
$\Gamma_0$ is the function given by \eqref{1.14a}, and $\vert 1)$ is the vector that is represented 
by the constant function identically equal to $1$. ${\bar K}_0$ and ${\bar K}_2$ are self-adjoint with 
respect to the scalar product $(\ \vert\ )$.
Let $\{ \vert e_n) \}$ be a basis in the Hilbert 
space, and assume that $\bar{\cal C}$ has a spectral representation
\be
\bar{\cal C} = \sum_n \mu_n\,\vert e_n)( e_n\vert
\label{2.3}
\ee
with eigenvalues $\mu_n$.
The kernels ${\bar K}_{0,1,2}$ scale with different powers of the temperature, see \eqref{1.13},
which facilitates a low-temperature expansion. In order to easily compare coefficients,
we write
\be
{\bar K} = {\bar K}_0 + \alpha {\bar K}_1 - \alpha^2 {\bar K}_2\ ,
\label{2.4}
\ee
treat $\alpha$ as a small parameter, and put $\alpha=1$ in the end. Now consider the eigenproblem
\be
\bar{\cal C} \vert e_n) = \mu_n \vert e_n)
\label{2.5}
\ee
and expand the eigenvalues and eigenvectors in powers of $\alpha$:
\bse
\label{2.6}
\bea
\mu_n &=& \mu_n^{(0)} + \alpha\,\mu_n^{(1)} + \alpha^2\,\mu_n^{(2)} + O(\alpha^3)
\label{2.6a}\\
\vert e_n) &=& \vert e_n^{(0)}) + \alpha \vert e_n^{(1)}) + \alpha^2 \vert e_n^{(2)}) + O(\alpha^3)\ .
\label{2.6b}
\eea
\ese
$\bar{\cal C}$ will be dominated by the eigenvalue with the smallest absolute value, which we denote by $\mu_0$.
To zeroth order in $\alpha$ we have, from the definition of $\Gamma_0$, \eqref{1.11},
\bse
\label{2.7}
\bea
\mu_0^{(0)} &=& 0\ ,
\label{2.7a}\\
\vert e_0^{(0)}) &=& \vert 1)\ .
\label{2.7b}
\eea
\ese
To linear order in $\alpha$ we have
\bse
\label{2.8}
\be
\bar{\cal C}_0\vert e_0^{(1)}) + {\bar K}_1 \vert 1) = \mu_0^{(1)}\vert 1)\ ,
\label{2.8a}
\ee
where $\bar{\cal C}_0 = {\bar K}_0 - \Gamma_0\mathbbm{1}$.
Multiplying from the left with $(1\vert$, and using $( 1\vert \bar{\cal C}_0 = 0$, we obtain
\be
\mu_0^{(1)} = 0
\label{2.8b}
\ee
since $( 1 \vert {\bar K}_1 \vert 1) = 0$ by symmetry. For the corresponding eigenvector we have the formal expression
\be
\vert e_0^{(1)}) = -\bar{\cal C}_0^{-1} {\bar K}_1 \vert 1) = -\bar{\cal C}_0^{-1}\vert\Gamma_1)
\label{2.8c}
\ee
where we have used \eqref{1.14b}. The inverse operator 
\be
\bar{\cal C}_0^{-1} = \frac{1}{\Gamma_0} (K_0 - \mathbbm{1})^{-1} = \frac{-1}{\Gamma_0}\sum_{m=0}^{\infty} (K_0)^m
\label{2.8d}
\ee
\ese
exists in this context in a formal sense since the zero eigenvalue does not contribute by symmetry. 
To second order in $\alpha$ we have
\be
\bar{\cal C}_0 \vert e_0^{(2)}) + {\bar K}_1 \vert e_0^{(1)}) - {\bar K}_2 \vert 1) = \mu_0^{(2)} \vert 1)\ .
\label{2.9}
\ee
Again multiplying from the left with $( 1\vert$ again we find
\bea
(1\vert 1)\,\mu_0^{(2)} &=& - ( 1\vert {\bar K}_2\vert 1) +( 1 \vert {\bar K}_1 \vert e_0^{(1)})
\nonumber\\
&=& - ( 1\vert \Gamma_2) + ( \Gamma_1\vert  \bar{\cal C}_0^{-1} \vert \Gamma_1)
\nonumber\\
&=& - ( 1 \vert \Gamma_2) - \sum_{m=0}^{\infty} ( \Gamma_1/\Gamma_0 \vert (K_0)^m  \vert \Gamma_1)
\label{2.10}
\eea
where we have used \eqref{2.8c}, \eqref{1.14}, and Lemma~\ref{lemma:2.1}. If we denote the average of functions $\Gamma$ with respect to the
weight $w$ by $\langle\Gamma\rangle_w = \intR d\epsilon\,w(\epsilon)\,\Gamma(\epsilon)$, and use Lemma~\ref{lemma:2.1} again, this can be written
\be
\mu_0^{(2)} = -\langle\Gamma_2\rangle_w - \sum_{m=0}^{\infty} \langle\Gamma_1 (K_0)^m \Gamma_1/\Gamma_0\rangle_w\ .
\label{2.11}
\ee
For the solution of \eqref{2.2} to lowest order in $\alpha$, or $T$, we now have
\bse
\label{2.12}
\be
\vert\varphi) = \phi \vert 1)
\label{2.12a}
\ee
with
\be
\phi = -1/\mu_0^{(2)}\ .
\label{2.12b}
\ee
\ese
For the conductivity, $\sigma(T) \propto (1\vert\varphi)$, this implies
\be
\sigma(T) \propto \phi = -1/\mu_0^{(2)}\ .
\label{2.13}
\ee
We now recall the temperature scaling of the various quantities. The rates $\Gamma_{0,1,2}$ scale as
$\Gamma_{\nu} \sim T^{\nu+3}$, see \eqref{1.14}, and the kernel $K_0$ scales as $T^0$. Accordingly,
\be
\sigma(T\to 0) \propto T^{-5}
\label{2.14}
\ee
to lowest order in $T$.
\begin{remark}
With the help of \eqref{1.14} one sees explicitly that the second term in \eqref{2.11} is small
compared to the first term by a factor of $(\omega_0/\epsilonF)^2$: $\langle\Gamma_2\rangle_w \propto T^5/\omega_0^4$,
while  $\langle\Gamma_1 (K_0)^m \Gamma_1/\Gamma_0\rangle_w \propto T^5/\omega_0^2 \epsilonF^2$.
\end{remark}
In the next section we show that the scheme outlined above, with minor modifications, can be made rigorous.

\section{Rigorous solution of the integral equation}
\label{sec:3}

\subsection{Preliminaries}
\label{subsec:3.1}

We list some useful properties of the kernels and the weight function $w$. The proofs are by
means of explicit elementary calculations.
\begin{lemma}
The function $\gamma_0$ defined in \eqref{1.14a} is even, positive definite, increases monotonically from its minimum at $\epsilon=0$, and
\be
\gamma_0(x\to\infty) \propto x^3\ .
\label{3.1}
\ee
\label{lemma:2.2}
\end{lemma}
\begin{lemma}
The distribution functions in the definition of the ${\bar K}_{\nu}$, \eqref{1.9b}, can be written
\be
n\left(\frac{u-\epsilon}{T}\right) + f\left(\frac{u}{T}\right) = \frac{\sqrt{w(u)/w(\epsilon)}}{2\sinh((u-\epsilon)/2T)}\ ,
\label{3.2}
\ee
and hence ${\bar K}_{0,2}(\epsilon,u)\geq 0$, and ${\bar K}_{0,2}(\epsilon,u) = 0$ if and only if $\epsilon=u$.
\label{lemma:3.2}
\end{lemma}

\subsection{Properties of the kernel $K_0$}
\label{subsec:3.2}

We now investigate the properties of the kernel $K_0$. Consider the space $L^2$ of real-valued square-integrable functions $\hat\varphi$ with
scalar product
\be
(\hat\varphi,\hat\psi) = \intR d\epsilon\,\hat\varphi(\epsilon)\hat\psi(\epsilon)\ ,
\label{3.3}
\ee
i.e., $(\hat\varphi,\hat\varphi) < \infty \ \forall \hat\varphi\in L^2$. In $L^2$ we define an operator with kernel
\be
\hat K_0(\epsilon,u) = \sqrt{w(\epsilon)\Gamma_0(\epsilon)}\,K_0(\epsilon,u)/\sqrt{w(u)\Gamma_0(u)}\ .
\label{3.4}
\ee
\begin{proposition}
${\hat K}_0$ is compact, or completely continuous, in $L^2$.
\label{proposition:3.0}
\end{proposition}
\begin{proof}
We show that $\hat K_0$ is Hilbert-Schmidt, or square integrable. Using \eqref{3.4}, \eqref{1.9}, and \eqref{3.2}, and omitting positive constants, we have
\begin{eqnarray*}
\intR d\epsilon\,du\,\left({\hat K}_0(\epsilon,u)\right)^2 &=& \intR d\epsilon\intR du\,\frac{\left({\bar V}_0''(\epsilon-u)\right)^2}{4\sinh^2((\epsilon -u)/2T)}\,\frac{1}{\Gamma_0(\epsilon)\Gamma_0(u)}
\nonumber\\
&\propto& \intR dx\intR dy\,\frac{(x-y)^4}{\sinh^2((x-y)/2)}\,\frac{1}{\gamma_0(x)\gamma_0(y)}
\nonumber\\
&\propto& \intR dx\intR dy\,\frac{x^4}{\sinh^2 x}\,\frac{1}{\gamma_0(2x+2y)\gamma_0(2y)}
\nonumber\\
&\leq& \frac{1}{\gamma_0(0)} \intR dx\,\frac{x^4}{\sinh^2 x} \intR dy\,\frac{1}{\gamma_0(2y)} 
\nonumber\\
&<& \infty
\end{eqnarray*}
where we have used Lemma~\ref{lemma:2.2} to obtain the upper bound. 
This proves compactness, since every Hilbert-Schmidt operator is compact by \cite{Akhiezer_Glazman_1993} p.58.
\qed
\end{proof}

Now consider the space $L^2_{\sigma}$ of real valued functions $\varphi$, $\psi$ that are square integrable with respect to 
the measure $d\sigma(\epsilon)$, where
\bse
\label{3.5}
\be
\sigma(\epsilon) = \int_{-\infty}^{\epsilon} du\,w(u)\Gamma_0(u)
\label{3.5a}
\ee
is the non-decreasing distribution function that characterizes the measure. The scalar product is
\be
\langle\varphi\vert\psi\rangle = \intR d\sigma(\epsilon)\,\varphi(\epsilon)\psi(\epsilon)\ ,
\label{3.5b}
\ee
and the vector norm is
\be
\vert\vert\varphi\vert\vert = \langle\varphi\vert\varphi\rangle^{1/2}\ .
\label{3.5c}
\ee
\ese
\begin{remark}
$L^2$ and $L^2_{\sigma}$ are both complete with a norm derived from
the scalar product, and hence are Hilbert spaces  (\cite{Akhiezer_Glazman_1993} Ch.1). The relation
\bse
\label{3.6}
\be
\hat\varphi(\epsilon)\, \hat= \, \varphi(\epsilon) \sqrt{w(\epsilon)\Gamma_0(\epsilon)}
\label{3.6a}
\ee
provides an isomorphism between $L^2$ and  $L^2_{\sigma}$. The relation
\be
\hat K(\epsilon,u)\, \hat=\, \sqrt{w(\epsilon)\Gamma_0(\epsilon)}\,K(\epsilon,u)/\sqrt{w(u)\Gamma_0(u)}
\label{3.6b}
\ee
\ese
provides an isomorphism between operators in $L^2$ and operators in $L^2_{\sigma}$.
\label{remark:3.1}
\end{remark}
\begin{proposition}
$K_0$ is self-adjoint with respect to the scalar product \eqref{3.5} in $L^2_{\sigma}$ 
and has a spectral representation
\be
K_0 = \sum_n \lambda_n\,\vert e_n \rangle \langle e_n\vert\ .
\label{3.7}
\ee
The spectrum contains an eigenvalue $\lambda_0 = 1$ of multiplicity one and an associated normalized eigenvector $\vert e_0\rangle$ given by
the constant function $e_0(\epsilon)\equiv 1/\langle 1\vert 1 \rangle$. The remainder of the spectrum
is purely discrete with eigenvalues  $\lambda_1, \lambda_2, \ldots$ such that $\vert\lambda_n\vert < 1$ $\forall n\geq 1$ 
and orthonormal eigenvectors $\vert e_n\rangle$.  $\lambda_{\infty}=0$ is the only accumulation point of eigenvalues. 
\label{proposition:3.1}
\end{proposition}
\begin{proof}
The self-adjointness of $K_0$ follows from Lemma \ref{lemma:2.1}, and $\lambda_0 = 1$ is an eigenvalue with eigenfunction 
$\vert e_0\rangle \propto \vert 1\rangle$, where $\vert 1\rangle$ is the vector that represents the constant function $e(\epsilon) \equiv 1$,
by \eqref{1.9}, \eqref{1.11}.  Since $\hat K_0$ is compact in $L^2$ by Proposition \ref{proposition:3.0}, $K_0$ is compact in $L^2_{\sigma}$ by the isomorphism \eqref{3.6b}. 
The existence of a spectral representation and the discreteness of the spectrum then follow from \cite{Schatten_1960} Theorem 6.
It remains to be shown that the eigenvalue $\lambda_0=1$ has multiplicity one, and that $\vert\lambda_{n\geq1}\vert<1$. To this end we write the eigenproblem
\bea
\lambda \varphi(\epsilon) &=& \intR du\,K_0(\epsilon,u)\,\varphi(u)
\nonumber\\
                              &=& \intR du\,K_0(\epsilon,u)\,\varphi(u) -  \varphi(\epsilon) +\varphi(\epsilon)
\nonumber\\                              
                              &=&\varphi(\epsilon) + ({\cal C}_0\varphi)(\epsilon)
\label{3.8}
\eea 
with
\be
({\cal C}_0\varphi)(\epsilon) = \intR du\,K_0(\epsilon,u+\epsilon) \left[\varphi(u+\epsilon) - \varphi(\epsilon)\right]\ .
\label{3.9}
\ee
Rearranging the identity \eqref{3.8}, multiplying by $w(\epsilon)\Gamma_0(\epsilon)$, and integrating over $\epsilon$ yields
\bse
\label{3.10}
\bea
(1-\lambda) \intR d\epsilon\,w(\epsilon) \Gamma_0(\epsilon) \left(\varphi(\epsilon)\right)^2  &=& 
\nonumber\\
&& \hskip -110pt = - \intR d\epsilon\,du\,w(\epsilon)\,
     {\bar K}_0(\epsilon,u+\epsilon)\,\varphi(\epsilon) \left[\varphi(u+\epsilon) - \varphi(\epsilon)\right]
\nonumber\\
&& \hskip -110pt = - \intR d\epsilon\,du\,w(u+\epsilon)\, {\bar K}_0(u+\epsilon,\epsilon)\,\varphi(\epsilon) \left[\varphi(u+\epsilon) - \varphi(\epsilon)\right]
\nonumber\\
 && \hskip -110pt = \intR d\epsilon\,du\,w(\epsilon)\,{\bar K}_0(\epsilon,u+\epsilon)\,\varphi(u+\epsilon) \left[\varphi(u+\epsilon) - \varphi(\epsilon)\right]
 \nonumber\\
 && \hskip -110pt = \frac{1}{2} \intR d\epsilon\,du\,w(\epsilon)\,{\bar K}_0(\epsilon,u+\epsilon) \left[\varphi(u+\epsilon) - \varphi(\epsilon)\right]^2\ .
\label{3.10a}
\eea          
Here we have used \eqref{2.1a} in going from the second line to the third line. In going from the third line to the fourth line we have
used Fubini's theorem to interchange the order of integration. Finally we have added the third and fourth lines to arrive at the last line. 
The right-hand side of \eqref{3.10a} is positive semi-definite by Lemma~\ref{lemma:3.2}. We conclude that $\lambda = 1$ if and only if
$\varphi(u+\epsilon) = \varphi(\epsilon)$ $\forall \epsilon,u$, which implies that
$\varphi$ is the constant function. There is thus only one linearly independent eigenfunction for the eigenvalue $\lambda=1$. Furthermore, all other 
eigenvalues obey $\lambda<1$.  Analogous reasoning yields
\bea
(1+\lambda) \intR d\epsilon\,w(\epsilon) \Gamma_0(\epsilon) \left(\varphi(\epsilon)\right)^2  &=& 
 \nonumber\\
 && \hskip -110pt = \frac{1}{2} \intR d\epsilon\,du\,w(\epsilon)\,{\bar K}_0(\epsilon,u+\epsilon) \left[\varphi(u+\epsilon) + \varphi(\epsilon)\right]^2\ .
\label{3.10b}
\eea      
\ese
$\lambda= -1$ thus implies $\varphi(u+\epsilon)+\varphi(\epsilon)=0$ $\forall \epsilon,u$, which implies that $\varphi$ is
the null function. $\lambda=-1$ therefore is not an eigenvalue. All eigenvalues must obey $\lambda>-1$ in addition to the
condition $\lambda<1$, and hence $-1<\lambda_n<1$ $\forall n>0$.
\qed
\end{proof}
\begin{remark}
Since the spectrum of $K_0$ is bounded by the operator norm $\vert\vert K_0\vert\vert$, this result also implies $\vert\vert K_0\vert\vert = 1$. 
\end{remark}
\begin{remark}
The upper bound in the proof that $\hat K_0$ is Hilbert-Schmidt remains valid if we replace $\bar V_0''$ by $\bar V_2''$, and from 
Lemma \ref{lemma:2.1} it follows that the kernel $K_2$ is self-adjoint. Consequently, the symmetric part of the full kernel, 
$K_+ = K_0 - K_2$ has a purely discrete spectrum and a spectral representation. 
\end{remark}
\begin{remark}
${\cal C}_0$ has the structure of a collision operator in kinetic theory \cite{Cercignani_1988}.
\end{remark}

\begin{corollary}
The collision operator ${\cal C}_0$ has a spectral representation in terms of the eigenvectors $\vert e_n\rangle$ of $K_0$,
\be
{\cal C}_0 = \sum_n \tilde\lambda_n \vert e_n\rangle\langle e_n\vert
\label{3.11}
\ee
with a purely discrete spectrum with one zero eigenvalue $\tilde\lambda_0 = 0$ of multiplicity one, and all other eigenvalues are negative
in the interval $\tilde\lambda_{n\geq 1} \in \ ]\!\! -2,0[$.
\label{corollary:3.1}
\end{corollary}
\begin{proof}
From the definition of ${\cal C}_0$, \eqref{3.9}, we have $K_0 - \mathbbm{1}$, with $\mathbbm{1}$ the identity operator in $L_{\sigma}$.
Since $L_{\sigma}$ is complete we have $\mathbbm{1} = \sum_n \vert e_n\rangle\langle e_n\vert$, and the corollary follows from
Proposition \ref{proposition:3.1}, with $\tilde\lambda_n = \lambda_n - 1$.
\qed
\end{proof}

\subsection{Temperature scaling of the integral equation}
\label{subsec:3.2}

We use \eqref{3.8} to write the integral equation \eqref{1.7} in terms of the collision operator ${\cal C}_0$. 
If we write $\varphi$ in terms of its even and odd parts $\varphi_{\pm}(\epsilon) = [\varphi(\epsilon) \pm \varphi(-\epsilon)]/2$,
\eqref{1.7} takes the form of two coupled integral equations
\bse
\label{3.12}
\bea
{\cal C}_0\vert\varphi_+\rangle &=& -\vert 1/\Gamma_0\rangle - K_1 \vert\varphi_-\rangle + K_2\vert\varphi_+\rangle
\label{3.12a}\\
{\cal C}_0\vert\varphi_-\rangle &=& \hskip 37pt                         - \ K_1\vert\varphi_+\rangle + K_2\vert\varphi_-\rangle\ .
\label{3.12b}
\eea
\ese
In what follows we say that a function $f$ scales as $T^n$, $f\sim T^n$, if $f \propto T^n \phi(\epsilon/T)$. Let $\varphi_{\pm} \sim T^{-n_{\pm}}$. 
Then from \eqref{1.13}, \eqref{1.14} it follows that ${\cal C}_0\vert\varphi_{\pm}\rangle \sim T^{-n_{\pm}}$, $\vert 1/\Gamma_0\rangle \sim T^{-3}$,
$K_1\vert\varphi_{\pm}\rangle \sim T^{-n_{\pm}+1}$, and $K_2\vert\varphi_{\pm}\rangle \sim T^{-n_{\pm}+2}$.
\begin{lemma}
For $T\to 0$, the leading contribution to $\varphi_+$ is a constant function, $\vert\varphi_+\rangle = \phi_+ \vert 1\rangle$, 
and $n_+ = n_- + 1 =5$.
\label{lemma:3.3}
\end{lemma}
\begin{proof}
$\varphi_-$ cannot be a constant function other than the null function, in which case \eqref{3.12b} would have no solution
for $\varphi_+$ other than the null function.
Therefore, the second term on the righ-hand side of \eqref{3.12b} is negligible for $T\to 0$ and we have $-n_- = -n_+ + 1$.
The last two terms on the right-hand side of \eqref{3.12a} therefore scale the same way. If $\varphi_+$ were not constant,
then these two terms would be negligible compared to the other two terms in \eqref{3.12a}. But the resulting equation,
${\cal C}_0\vert\varphi_-\rangle = - \vert 1/\Gamma_0\rangle$, has no solution since the inhomogeneity is not orthogonal
to the zero eigenvector of ${\cal C}_0$, $\langle e_0\vert 1/\Gamma_0\rangle \neq 0$. Therefore, the leading contribution
to $\varphi_+$ as $T\to 0$ must be a constant function, and $n_+ = 5$ follows from the requirement that the two
$\varphi$ dependent terms on the right-hand side of \eqref{3.12a} have the same scaling behavior as the inhomogeneity
$\vert 1/\Gamma_0\rangle \sim T^{-3}$.
\qed
\end{proof}

\begin{proposition}
To leading order for $T\to 0$, 
\be
\vert\varphi_-\rangle = -\phi_+ \sum_{m=0}^{\infty} \left( K_0\right)^m \vert\Gamma_1/\Gamma_0\rangle\ .
\label{3.13}
\ee
\label{proposition:3.3}
\end{proposition}
\begin{proof}
$K_2\vert\varphi_-\rangle$ in \eqref{3.12b} is of relative order $T^2$ compared to ${\cal C}_0\vert\varphi_-\rangle$, and with 
Lemma~\ref{lemma:3.3} and Corollary~\ref{corollary:3.1} we have
\be
{\cal C}_0\vert\varphi_-\rangle = \sum_n (1-\lambda_n)\vert e_n\rangle\langle e_n\vert \varphi_-\rangle = -\phi_+ K_1\vert 1\rangle = -\phi_+\vert\Gamma_1/\Gamma_0\rangle\ .
\label{3.14}
\ee
We solve for $\vert\varphi_-\rangle$ by multiplying from the left with $\sum_{m\neq 0} \frac{1}{1-\lambda_m}\,\vert e_m\rangle\langle e_m \vert$
and obtain
\be
\vert\varphi_-\rangle = -\phi_+ \sum_{n\neq 0} \frac{1}{1-\lambda_n}\,\vert e_n\rangle\langle e_n\vert\Gamma_1/\Gamma_0\rangle\ .
\label{3.15}
\ee
Now let $f_-$ be an odd function in $L^2_{\sigma}$. From \eqref{3.7} and the fact that the eigenfunction $e_0$ is even we have
\be
\sum_{m=0}^{\infty} \left(K_0\right)^m \vert f_-\rangle = \sum_{m=0}^{\infty} \sum_{n\neq 0} \left(\lambda_n\right)^m \langle e_n\vert f_-\rangle \vert e_n\rangle\ .
\label{3.16}
\ee
To show that this series converges we use the fact that $\vert\lambda_{n\neq 0}\vert<1$ by Proposition~\ref{proposition:3.1}. 
Let $\lambda_*^2 = \sup\{\lambda_n^2 \, ; n\neq 0\}$. Then 
\bea
\left\vert\left\vert \sum_{m=0}^{\infty} \left(K_0\right)^m \vert f_-\rangle \right\vert\right\vert^2 &=& 
                                                      \sum_{m=0}^{\infty} \sum_{n\neq 0} \left(\lambda_n\right)^{2m} \langle e_n\vert f_-\rangle^2
< \sum_{m=0}^{\infty} \left(\lambda_*^2\right)^{m} \sum_{n\neq 0} \langle e_n\vert f_-\rangle^2
\nonumber\\
&=& \frac{1}{1-\lambda_*^2} \sum_n   \langle e_n\vert f_-\rangle^2
= \frac{1}{1-\lambda_*^2}\,\vert\vert f_-\vert\vert^2 < \infty\ .
\label{3.17}
\eea
where we have used $\langle e_0\vert f_-\rangle = 0$. We therefore can interchange the summations in \eqref{3.16} and obtain
\be
\sum_{m=0}^{\infty} \left(K_0\right)^m \vert f_-\rangle =  \sum_{n\neq 0} \frac{1}{1-\lambda_n}\,\langle e_n\vert f_-\rangle \vert e_n\rangle\ .
\label{3.18}
\ee
Putting $\vert f_-\rangle = \vert\Gamma_1/\Gamma_0\rangle$ and using \eqref{3.18} in \eqref{3.15} shows that \eqref{3.13} is
a solution of \eqref{3.14}, and hence of \eqref{3.12b}, to lowest order in $T$. The most general solution is obtained by adding
the general solution of the corresponding homogeneous equation. However, by Corollary~\ref{corollary:3.1} the latter,
${\cal C}_0\vert\varphi_-\rangle = 0$, has only the null solution since the zero eigenfunction $e_0$ is even. Hence the solution
\eqref{3.13} is unique, which proves the proposition.
\qed
\end{proof}

We now are in a position to prove our final result.
\begin{theorem}
The transport relaxation time $\tau(T)$ in \eqref{1.1} obeys the Bloch law $\tau(T\to 0) \propto 1/T^5$ and is given, to leading
order in $T$, by $\tau(T) = \phi_+/2$ with
\be
\phi_+ = \frac{1}{\langle\Gamma_2\rangle_w + \sum_{m=0}^{\infty} \langle\Gamma_1 (K_0)^m \Gamma_1/\Gamma_0\rangle_w}\ ,
\label{3.19}
\ee
where $\langle \ldots \rangle_w = \intR d\epsilon\, \ldots w(\epsilon)/\intR d\epsilon\,w(\epsilon)$
denotes an average with weight $w$. 
\label{theorem:3.1}
\end{theorem}
\begin{proof}
Since  $w(\epsilon)$ is even, only $\varphi_+$ contributes to the conductivity, and by Lemma~\ref{lemma:3.3} $\varphi_+$,
to leading order for $T\to 0$, is a constant function that scales as $T^{-5}$. To determine the prefactor we multiply
\eqref{3.12a} from the left with the zero eigenfunction $\langle e_0\vert$. This yields
\bea
0 &=& -\langle e_0\vert 1/\Gamma_0\rangle - \langle e_0\vert K_1\vert\varphi_-\rangle + \phi_+ \langle e_0\vert K_2\vert 1\rangle
\nonumber\\
   &=& -\langle 1\vert 1/\Gamma_0\rangle - \langle \Gamma_1/\Gamma_0 \vert\varphi_-\rangle + \phi_+\langle \Gamma_2/\Gamma_0 \vert 1\rangle
\nonumber\\
   &=&  -\langle 1\vert 1/\Gamma_0\rangle + \phi_+ \sum_{m=0}^{\infty} \langle\Gamma_1/\Gamma_0 \vert (K_0)^m \vert\Gamma_1/\Gamma_0\rangle
           + \phi_+ \langle 1 \vert \Gamma_2/\Gamma_0\rangle
\label{3.20}
\eea
where we have used Proposition~\ref{proposition:3.3}. Solving for $\phi_+$ yields \eqref{3.19}.
\end{proof}
\begin{remark}
The formal solution in Sec.~\ref{sec:2} correctly reproduces the rigorous solution given above.
\label{remark:3.5}
\end{remark}

\section{Conclusion}
\label{sec:4}

We have provided a rigorous solution of the kinetic equation that describes the scattering of electrons by acoustic
phonons, under the Bloch assumption of phonon equilibrium, in the limit of low temperature. This kinetic equation
can be derived either from the linearized Boltzmann equation or from the Kubo formula. Our result establishes the
Bloch $T^5$ law for the electrical resistivity in three spatial dimensions as the exact solution of the Boltzmann equation 
at asymptotically low temperatures. We conclude with some additional remarks.
\begin{remark}
Our solution of the Boltzmann equation is rigorous. The Kubo formula yields the same result; however, no
rigorous derivation of \eqref{1.7} from the Kubo formula has been provided. The calculations that lead to
\eqref{1.7} from the full Kubo formula turn out to be exact to leading order in the scattering potential,
however, this fact has only been established by comparing with the solution of the Boltzmann equation.
See also Remarks~\ref{remark:1.9} and \ref{remark:1.9.1}.
\label{remark:4.1}
\end{remark}
\begin{remark}
The symmetry properties of the kernel $K_1$, which is skew-adjoint with respect to the scalar product $\langle\ \vert\ \rangle$,
ensure that its contribution to the scattering rate is proportional to $T^5/\omega_0^2\epsilonF^2$, rather than $T^4/\omega_0^2\epsilonF$
as suggested by naive scaling arguments based on \eqref{1.13} or \eqref{1.14}.
\label{remark:4.2}
\end{remark}
\begin{remark}
The skew-adjointness of $K_1$ is a result of rewriting the underlying integral equation for the rate $\phi$ defined in
\eqref{1.4a} into one for the inverse rate $\varphi$ defined in \eqref{1.4b} (recall that the collision operator $C$ in
\eqref{1.6} was self-adjoint). As another consequence, the zero eigenvalue of ${\cal C}_0$ reflects the approximate
momentum conservation in the limit $T\to 0$ mentioned in Remark~\ref{remark:1.5}.
\label{remark:4.3}
\end{remark}
\begin{remark}
The simple approximation that turns \eqref{1.7} into an algebraic equation by replacing the rates $\Gamma_{1,2,3}$ 
by their averages with respect to $w$ \cite{Mahan_2000} correctly yields the Bloch law, and reproduces the first term in the
denominator of \eqref{3.19}, but obviously misses the second term since $\langle\Gamma_1\rangle_w=0$.
\label{remark:4.4}
\end{remark}
\begin{remark}
One can formally continue the expansion in powers of $\alpha$ in Sec.~\ref{sec:2}. Symmetry ensures that only even powers 
of $\alpha$ contribute to the eigenvalue, and the leading correction to the Bloch law scales as $T^{-3}$.
\label{remark:4.5}
\end{remark}
\begin{remark}
Scattering by excitations other than phonons can be analyzed analogously. Consider propagating (i.e., particle-like)
excitations with an en\-er\-gy-momentum relation $\omega_{\bm q} \propto \vert{\bm q}\vert^{\alpha}$ that lead to an
effective potential with spectrum
\be
V''({\bm p},u) \propto (\omega_{\bm p})^{\beta} \left[\delta(\omega_{\bm p} - u) - \delta(\omega_{\bm p} + u)\right]
\label{4.1}
\ee
with $\alpha, \beta > 0$ (for phonons, $\alpha = \beta = 1$). The relaxation rates then scale as
\bse
\label{4.1}
\bea
\Gamma_0 &\sim& T^{2/\alpha + \beta}
\label{4.1a}\\
\Gamma_1 &\sim& T^{2/\alpha + \beta + 1}
\label{4.1b}\\
\Gamma_2 &\sim& T^{4/\alpha + \beta}
\label{4.1c}
\eea
\ese
From \eqref{2.11} or \eqref{3.19} it follows that $\Gamma_1$ provides the leading temperature
dependence of the conductivity if $\alpha<1$, and does not contribute to the leading behavior
if $\alpha >1$. The phonon case with $\alpha=1$ is the marginal one where $\Gamma_2$ and
$\Gamma_1$ both contribute to the leading behavior, albeit the latter with a small prefactor. 
\label{remark:4.6}
\end{remark}
\begin{remark}
An explicit example is provided by ferromagnetic
magnons, for which $\alpha=2$, $\beta=0$, but the spectrum has a gap due to the exchange splitting $T_0$.  For  $\vert u\vert>T_0$
one then has ${\bar V}_0''(u) \propto \sgn u$ \cite{Ueda_Moriya_1975,Bharadwaj_Belitz_Kirkpatrick_2014}. The relaxation rates
scale as $\Gamma_0 \sim T\ln T$, $\Gamma_1 \sim \Gamma_2 \sim T^2$. In agreement with the general argument
in Remark~\ref{remark:4.6}, $\Gamma_1$ does not contribute to the leading temperature dependence, and the conductivity, 
for $T>T_0$, behaves as $\sigma \propto 1/\Gamma_2 \propto T^{-2}$.
\label{remark:4.7}
\end{remark}
\begin{remark}
Another example is the case of Coulomb scattering, which leads to ${\bar V}_0''(u) \propto u$. This excitation is not particle-like;
it is characterized
by a continuous spectrum rather than a $\delta$-function as in the phonon and magnon cases. As a result the rates scale as
$\Gamma_0 \sim \Gamma_2 \sim T^2$, $\Gamma_1 \sim T^3$. $\Gamma_1$ again does not contribute to the
leading behavior, and for the conductivity one obtains the well-known Fermi-liquid result $\sigma \propto T^{-2}$. 
\label{remark:4.8}
\end{remark}
\begin{remark}
For both of the examples in Remarks~\ref{remark:4.7},~\ref{remark:4.8} the proof in Sec.~\ref{sec:3}, with appropriate modifications, 
still holds.  For the ferromagnetic case at asymptotically low temperatures, $T<T_0$, the exchange gap leads to complications that 
require a separate analysis.
\label{remark:4.9}
\end{remark}

\acknowledgement
We thank Bob Dorfman and Peng Lu for discussions. 
\bibliographystyle{spmpsci}      

%
%

\end{document}